\documentclass[doublecol]{epl2}

\usepackage{eprint}

\usepackage{amsthm,amsmath,amssymb}
\newtheorem*{thm}{Theorem}

\newcommand{\Prob}[1]{\mathbb{P}\!\left( #1 \right)}
\newcommand{\prob}[1]{\mathbb{P}_Q\!\left( #1 \right)}
\newcommand{\probL}[1]{\mathbb{P}_L\!\left( #1 \right)}
\newcommand{\pr}[1]{\mathbb{P}\!\left( #1 \right)}
\newcommand{\expec}[1]{\left<#1\right>}

\newcommand{\ket}[1]{ | #1 \rangle}
\newcommand{\bra}[1]{ \langle #1 |}

\newcommand{\Ac}{\mathcal{A}}
\newcommand{\Hc}{\mathcal{H}}
\newcommand{\Aci}{\mathcal{A}_{\infty}}
\newcommand{\Mc}{\mathcal{M}}

\newcommand{\la}{\lambda}

\def\one{\leavevmode\hbox{\small1\normalsize\kern-.33em1}}

\renewcommand{\revision}[1]{#1}

\begin{document}

\title{A \revision{tight Tsirelson} inequality for infinitely many outcomes}

\author{Stefan Zohren\inst{1,2,3}
\and Paul Reska\inst{4}
\and Richard D. Gill\inst{1,5}
\and Willem Westra\inst{6}}

\shortauthor{S. Zohren \etal}

\institute{
\inst{1} Mathematical Institute, Leiden University, Niels Bohrweg 1, 2333 CA Leiden, The Netherlands \\
\inst{2} Department of Statistics, S\~ao Paulo University, Rua do Mat\~ao 1010, 05508-090, S\~ao Paulo, Brazil\\
\inst{3} Blackett Laboratory, Imperial College London, Prince Consort Road, London SW7 2AZ, UK\\
\inst{4} Institute for Theoretical Physics, Utrecht University,  Leuvenlaan 4, 3584 CE Utrecht, The Netherlands\\
\inst{5} CWI, Science Park 123, 1098 XG Amsterdam, The Netherlands\\
\inst{6} Department of Mathematics, University of Iceland, Dunhaga 3, 107 Reykjavik, Iceland}

\date{December 20, 2009}

\pacs{03.65.Ud}{Entanglement and quantum nonlocality}
\pacs{03.65.-w}{Quantum mechanics}
\pacs{03.67.-a}{Quantum information}

\abstract{
We present a novel tight \revision{bound on the quantum violations of the CGLMP inequality} in the case of infinitely many outcomes. Like in the case of Tsirelson's inequality the proof of our new inequality does not require any assumptions on the dimension of the Hilbert space or kinds of operators involved. However, it is seen that the maximal violation is obtained by the conjectured best measurements and a pure, but not maximally entangled, state. We give an approximate state which, in the limit where the number of outcomes tends to infinity, goes to the optimal state for this setting. This state might be potentially relevant for experimental verifications of Bell inequalities through multi-dimenisonal entangled photon pairs.
}
\maketitle

\section{
Introduction}

Already since the seminal work by Bell in 1964 \cite{Bell}, Bell inequalities and their quantum violations are widely discussed in the literature. Probably one of the most well known examples is the Clauser-Horne-Shimony-Holt (CHSH) inequality \cite{CHSH}. It considers the case of two parties, Alice and Bob, which perform two possible measurements each with outcomes $\pm1$ on a shared quantum state. Any correlations of the experimental results which can be explained through a local realistic theory based on local hidden variables obey the CHSH inequality,
\begin{equation}\label{eq:CHSH}
|\expec{A_2B_2}+\expec{A_1B_2}+\expec{A_1B_1}-\expec{A_2B_1}|\leq 2,
\end{equation}
where $A_1,A_2$ and $B_1,B_2$ refer to the two possible measurements on Alice's and Bob's side respectively and $\expec{\cdot}$ denotes the expectation value.

Quantum correlations \revision{though} can violate this inequality. However, it was shown by Tsirelson \cite{Cir80} that they still obey the following \revision{so-called} \emph{quantum} Bell inequality
\begin{equation}
|\expec{A_2B_2}+\expec{A_1B_2}+\expec{A_1B_1}-\expec{A_2 B_1}|\leq 2 \sqrt{2},
\end{equation}
where the maximal violation is obtained by the maximally entangled state. This \revision{quantum Bell inequality or} Tsirelson inequality is quite remarkable in the sense that it applies to any quantum correlations without making assumptions on the kind of measurements or Hilbert space involved.

In this letter we give a proof of a \revision{tight} quantum Bell inequality analogous to Tsirelson's \revision{original} inequality for the case of infinitely many outcomes, \revision{that is a tight bound on the quantum violations of a generalization of the above Bell inequality for infinitely many outcomes}. In this limit, we show that the maximal violation is obtained by the conjectured best measurements \cite{CGLMP,Zukowski} and that the optimal state is a pure state which is not maximally entangled. This is in agreement with previous numerical investigations \cite{optimal,zohren,Chen}. Further, we give an explicit analytical expression for the state which, in the limit where the number of outcomes tends to infinity, converges to the optimal state for this setting. Numerical investigations show that for large but finite number of outcomes this state can be taken as a good approximation to the optimal state. Finally, we comment on possible experimental implementations. 

\section{
Basic definitions regarding Bell experiments}

Consider the setting described in the previous section of two parties, Alice and Bob, choosing between two possible measurements. However, we now generalize the measurements to have $d$ possible outcomes $x_k$, $k=0,...,d-1$ with $x_k\neq x_l$ for $k\neq l$. We call this the $2\times2\times d$ Bell setting. Expectation values of the kind used above can then be written as
\begin{equation}
\expec{A_aB_b}=\sum_{k,l=0}^{d-1} x_k x_l \Prob{k,l|a,b}\quad a,b=1,2.
\end{equation}
In the case of local realistic theories any correlation between the measurements on Alice's and Bob's side must be explained by a hidden variable $\la$ which implies for the probabilities
\begin{equation}\label{eq:LR}
\probL{k,l|a,b}= \sum_\la p(\la) \pr{k|a,\la} \pr{l|b,\la}.
\end{equation}

In quantum mechanics on the other hand the system is described by a density matrix $\rho$ on a Hilbert space $\Hc\!=\!\Hc_A\otimes\Hc_B$ and quantum mechanical probabilities read 
\begin{equation}\label{eq:QMprob}
\prob{k,l|a,b}= \mathrm{Tr}\left(A_a^k \otimes B_b^l \,\,\rho \right).
\end{equation}
Here $A_a^k$ and $B_b^l$ are \revision{positive} operators on $\Hc_A$ and $\Hc_B$ respectively, satisfying $\sum_{k=0}^{d-1}A_a^k =\one$ for all $a$ and $\sum_{l=0}^{d-1}B_b^l \,=\one$ for all $b$. \\

\section{A Bell inequality for the CGLMP setting}
A Bell type inequality for the $2\times 2 \times d$ setting was first given by Collins-Gisin-Linden-Massar-Popescu (CGLMP) \cite{CGLMP}. Later, based on earlier ideas \cite{optimal}, a generalized and simplified version was found in \cite{zohren} which reads
\begin{eqnarray}\label{eq:inequality}
\!\!\!\probL{A_2<B_2}+\probL{B_2<A_1}+\probL{A_1<B_1}+\nonumber\\
+\,\probL{B_1\leq A_2}\geqslant 1,\!\!\!\!\!
\end{eqnarray}
where
$\probL{A_a<B_b}=\sum_{k<l} \probL{k,l|a,b}$.
Not only the inequality itself became much simpler, but also its proof which reduces the proof of \cite{CGLMP} to a literally three line proof and is therefore worth to quickly recall here. Starting with the following obvious statement
$\{A_2\geq B_2\} \cap \{B_2 \geq A_1\} \cap \{A_1\geq B_1\} \subseteq \{A_2\geq B_1\}$
 and taking the complement one gets $\{A_2<B_1\}\subseteq \{A_2<B_2\} \cup  \{B_2<A_1\}\cup \{A_1<B_1\}$ which implies for the probabilities that
$\probL{A_2<B_1} =   1- \probL{A_2 \geq B_1}
 \leq   \probL{A_2<B_2} +  \probL{B_2<A_1} +\probL{A_1<B_1}
 $.
This completes the proof. A nice feature of inequality \eqref{eq:inequality} is that it reads the same for any number of outcomes. In particular, it is also valid as the number of outcomes becomes infinite. It is also straightforward to generalize the inequality and the proof to a $2\times N\times d$ setting. As a special example we note that for the case of $d\!=\!2$ possible outcomes $\{ x_0,x_1\}\!=\!\{-1,+1$\} inequality \eqref{eq:inequality} directly reproduces the CHSH inequality in its conventional form \eqref{eq:CHSH}.

\section{A quantum Bell inequality for infinitely many outcomes}
The quantum violation of the CGLMP inequality for various numbers of outcomes was investigated in several articles, most importantly \cite{Ac02,Chen,zohren}. In most of them it was assumed that the dimension of the Hilbert space is equal to the number of outcomes, i.e. $\Hc\!=\!\mathbb{C}^d\otimes\mathbb{C}^d$, numerical evidence supporting this assumption was given in \cite{zohren}. Further, for higher numbers of outcomes the analysis was \revision{purely} numerical. From the numerical evidence it was conjectured that the optimal measurements, causing the maximal violation, are given by the following projective measurements \cite{CGLMP,Zukowski} with projectors
\begin{eqnarray}
\ket{k}_{A,a} & = &\frac{1}{\sqrt{d}} \sum_{m=0}^{d-1} \exp\left(i\frac{2\pi}{d}m (k+\alpha_a) \right) \ket{m}_A \label{eq:bestmeasA},\\
\ket{l}_{B,b} & = & \frac{1}{\sqrt{d}} \sum_{n=0}^{d-1} \exp\left(i\frac{2\pi}{d}n (-l+\beta_b) \right) \ket{n}_B \label{eq:bestmeasB},
\end{eqnarray}
with $\alpha_1=0$, $\alpha_2=1/2$, $\beta_1=1/4$ and $\beta_2=-1/4$. Further it was assumed, without loss of generality, that the optimal state is pure and hence the density matrix can be written in terms of the Schmidt decomposition of this state
\begin{equation}\label{eq.pure}
\rho=\ket{\psi}\bra{\psi},\quad\ket{\psi}=\sum_{k=0}^{d-1}\la_k\ket{kk},
\end{equation}
with $\la_k\ge0$ and the normalization condition $\sum_{k=0}^{d-1}\la_k^2=1$. 

The maximal violation of the generalized version of the CGLMP inequality, Eq.\ \eqref{eq:inequality}, was investigated in \cite{zohren}. In addition to giving further numerical evidence for the conjectured best measurements \eqref{eq:bestmeasA}--\eqref{eq:bestmeasB}, as well as for the assumption that the dimension of the Hilbert space can be taken equal to the number of outcomes, the numerical analysis was extended to very large numbers of outcomes of the order of one million. It was seen that for such large numbers of outcomes the left-hand-side of inequality \eqref{eq:inequality} tends slowly towards zero. From this there was conjectured a quantum Bell inequality for infinitely many outcomes which we prove in the following.
 
\begin{thm}[Quantum Bell inequality]
For the number of outcomes $d\!\to\!\infty$ the minimal value of $\prob{A_2<B_2}+\prob{B_2<A_1}+\prob{A_1<B_1}+\prob{B_1\leq A_2}$ converges to zero. Hence, 
\begin{eqnarray}
\!\!\!\Aci:=\prob{A_2<B_2}+\prob{B_2<A_1}+\nonumber\\
+\,\prob{A_1<B_1}+\prob{B_1\leq A_2}\geqslant 0\!\!\! \label{eq.thm}
\end{eqnarray}
is a \emph{tight} quantum Bell inequality for the $2\times2\times\infty$ Bell setting.
\end{thm}

\begin{proof}
Since the the left-hand-side of \eqref{eq.thm} is obviously non-negative, one only has to show the tightness of the inequality. From the numerical analysis described above we expect this to be achieved by a quantum behavior with the conjectured best measurements \eqref{eq:bestmeasA}--\eqref{eq:bestmeasB} and a pure state \eqref{eq.pure} in the limit $d\to\infty$. By inserting \eqref{eq:QMprob}, \eqref{eq:bestmeasA}--\eqref{eq:bestmeasB} and \eqref{eq.pure} into \eqref{eq.thm} it can be shown that for finite number of outcomes the left hand-side of \eqref{eq.thm} reads \cite{zohren}
\begin{equation}
\Ac_d(\la)=2-\frac{1}{d}\sum_{k,l=0}^{d-1}\frac{\la_k\la_l}{\cos\left(\frac{\pi}{2d}(k-l)\right)}.
\end{equation}
Taking the limit $d\to\infty$ of this expression yields
\begin{equation}
\Aci (f)\!=\!\lim_{d\to\infty}\Ac_d(\la)\!=\!2-\Mc (f),
\end{equation}
with
\begin{equation}\label{eq:Mf}
\Mc (f):=\int_0^1\!\int_0^1\!\frac{f(x)f(y)}{\cos\left(\frac{\pi}{2}(x-y)\right)}\,dx\, dy,
\end{equation}
where the function $f(x)$ is non-negative and normalized according to $\int_0^1 f^2(x)dx=1$. We will now show that $\sup_f \Mc (f)=2$. 

From the insight of previous numerical investigations \cite{Chen,zohren,zohren2} we make the following ansatz
\begin{eqnarray}
f_\delta(x)&=&\frac{\pi^{1/4}\,2^{2\delta-1/2}\left( x(1-x) \right)^{\delta-1/2}}{\sqrt{\Gamma(1/2-2\delta)\Gamma(2\delta)\cos(2\pi \delta)}},\quad \delta > 0\nonumber\\
&=& \sqrt{\delta}\left( x(1-x) \right)^{\delta-1/2} +\mathcal{O}(\delta). \label{eq:statecont}
\end{eqnarray}
Inserting this into \eqref{eq:Mf} we get
\begin{equation}
\Mc (f)=I_\delta +\mathcal{O}(\delta),
\end{equation}
where we defined
\begin{equation}
I_\delta:=\int_0^1\!\int_0^1\!\frac{\left( x(1-x)y(1-y) \right)^{\delta-1/2}\delta}{\cos\left(\frac{\pi}{2}(x-y)\right)}dx dy.
\end{equation}
For every $0<\epsilon\leq 1/2$ and $0<\delta< 1/2$ one has
\begin{eqnarray}
I_\delta&\geq& 2 \int_0^\epsilon\!\int_{1-\epsilon}^1\!\frac{\left( x(1-x)y(1-y) \right)^{\delta-1/2}\delta}{\cos\left(\frac{\pi}{2}(x-y)\right)}dx dy\nonumber\\
&\geq& \frac{4}{\pi} \int_0^\epsilon\!\int_{0}^\epsilon\!\frac{( xy)^{\delta-1/2}\delta}{(x+y)}dx dy\nonumber\\
&=&\frac{\epsilon^{2\delta}}{\pi} \left\{\Psi(1/4-\delta/2)-\Psi (1/4+\delta/2)+ \right. \nonumber\\
&&+\,\Psi(3/4-\delta/2)-\Psi(3/4+\delta/2)+ \nonumber\\
&&+\left.2\pi\sec(\pi \delta)\right\},\label{proofend}
\end{eqnarray}
where $\Psi(z)=\Gamma'(z)/\Gamma(z)$ is the digamma function. For all $0<\epsilon\leq 1/2$ the last term in \eqref{proofend} goes to $2$ as $\delta\to 0$. Since also $\Mc (f) \leq 2$ by the non-negativity of the left-hand side of \eqref{eq.thm}, it follows from \eqref{proofend} that $\sup_f \Mc(f)=2$ and hence $\inf_f \Aci (f)=0$ which completes the proof. 
\end{proof}

\section{The approximate state}
The state \eqref{eq:statecont} causing the maximal violation of the Bell inequality \eqref{eq:inequality} in the limit $d\to\infty$ can be seen as a regularized version of the following approximate state for finite $d$
\begin{equation}\label{eq:approx}.
\quad\ket{\psi_d}\sim\sum_{k=0}^{d-1}\frac{\ket{kk}}{\sqrt{(k+1)(d-k)}},
\end{equation}
correctly normalized according to $||\psi_d||^2=1$. The quantum violation of the CGLMP inequality for this state was first investigated in \cite{Chen} and in \cite{zohren2} for the case of inequality \eqref{eq:inequality} for a large number of outcomes of order $10^6$. The minimal value of the left-hand-side of inequality \eqref{eq:inequality} as a function of the number of outcomes $d$ is shown in Fig.\ \ref{fig:min} for both the approximate state as well as the optimal state using the conjectured best measurement operators. 
\begin{figure}
\onefigure[width=3.4in]{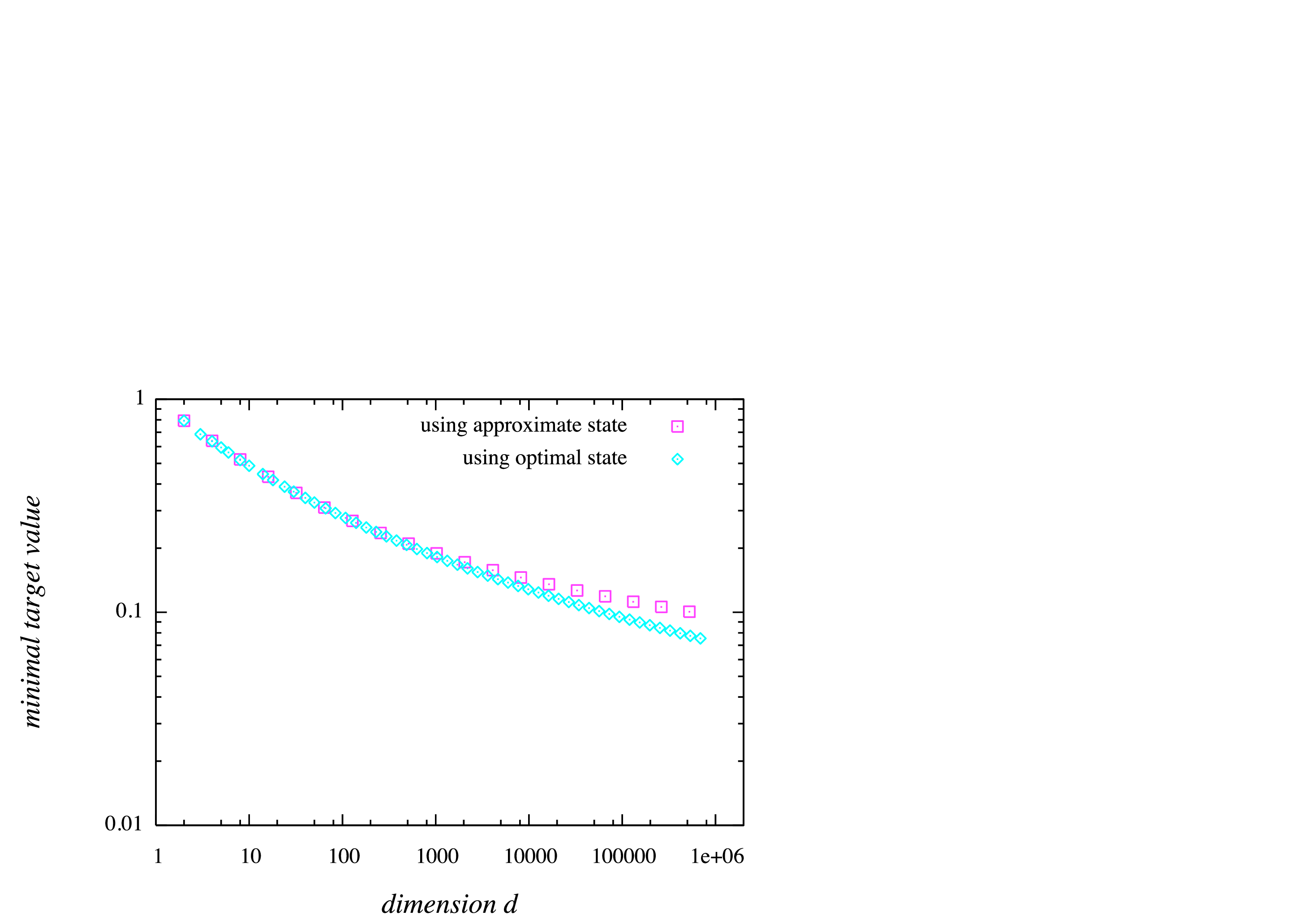}
\caption{Minimal value of the left-hand-side of inequality \eqref{eq:inequality}, corresponding to the maximal quantum violations of this inequality, as a function of the number of outcomes $d$ both for the optimal state as well as the approximate state \eqref{eq:approx}.}
\label{fig:min}
\end{figure}
In addition, Fig.\ \ref{fig:entropy} displays the entanglement entropy for both states as a function of the number of outcomes. One observes that in terms of maximal violation of the inequality the approximate state serves as a good approximation to the optimal state, even though the entanglement entropy is slightly different for large $d$. Nevertheless, the above results indicate that in the limit $d\to\infty$ both converge to the same state and hence should also have the same entanglement entropy. Using the definition of entanglement entropy of a pure state in terms of the Schmidt coefficients, i.e. $E(\psi)=-\sum_{i=0}^{d-1}\la_i^2\log\la_i^2$, we can calculate the entanglement entropy of the approximate state for $d\to\infty$ 
\begin{equation}
\lim_{d\to\infty}\frac{E(\psi_d)}{\log d}=\frac{1}{2}.
\end{equation}

\begin{figure}
\onefigure[width=3.4in]{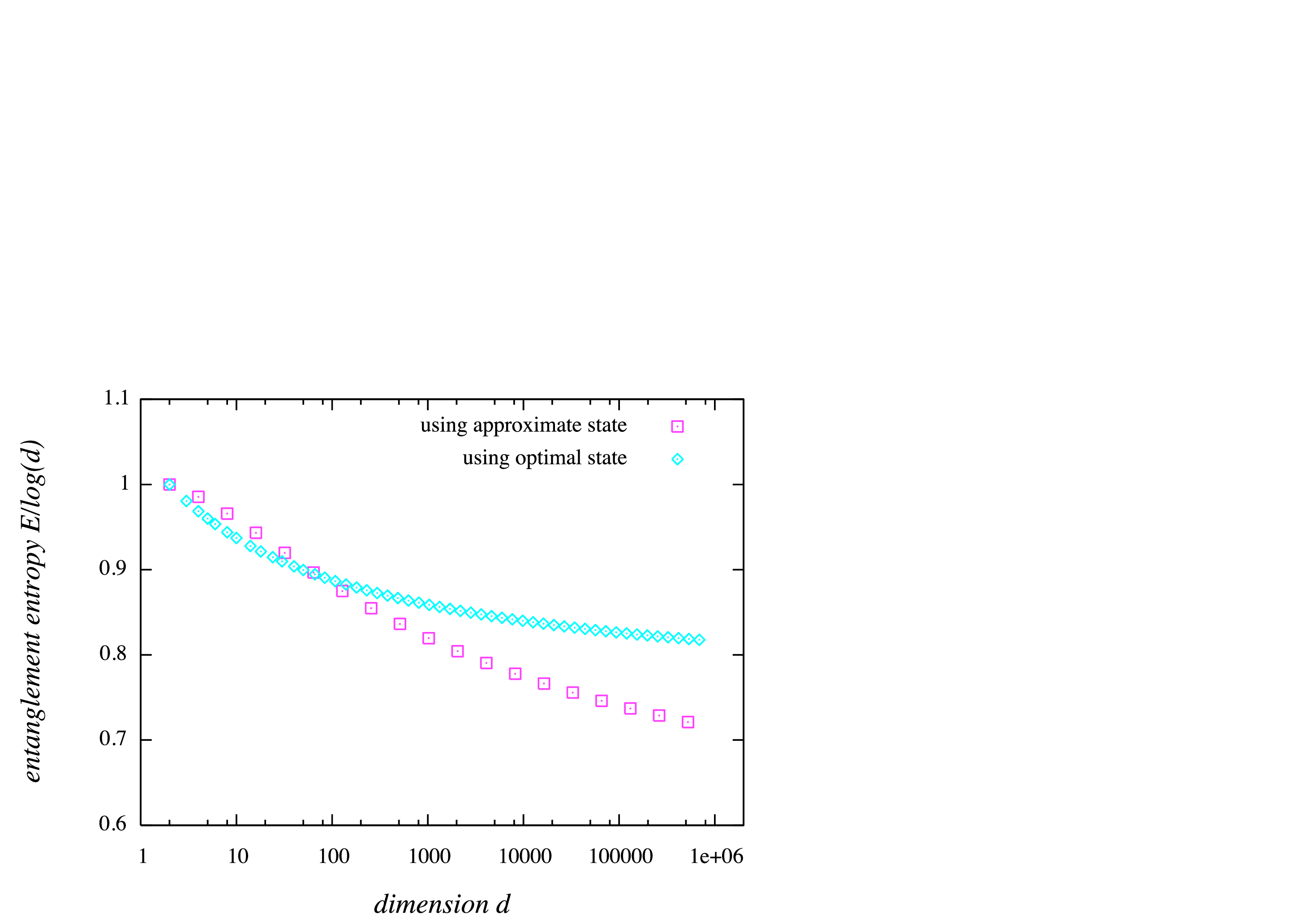}
\caption{Entanglement entropy of the optimal state and the approximate state as a function of the number of outcomes $d$.}
\label{fig:entropy}
\end{figure}

It is well known that for the number of outcomes $d\geq3$ the optimal state $\psi_d$ causing the maximal violation of the CGLMP inequality is not maximally entangled, i.e. $E(\psi_d)/\log d<1$ (see for example \cite{Ac02,optimal,zohren}). With respect to this it is interesting to see that in the case of $d\to\infty$ even an entanglement entropy  of $1/2$ is sufficient to cause the maximal violation. In addition, this result points further to the differences between maximal violation of Bell inequalities and the statistical strength of Bell experiments defined through Kullback-Leibler divergence or relative entropy \cite{statistical}. The latter was investigated in \cite{optimal} for the CGLMP inequality, where the entanglement entropy of the optimal state $\Psi_d$, which minimizes the Kullback-Leibler divergence, approached the asymptotic value $\lim_{d\to\infty}E(\Psi_d)/\log d\approx 0.69$. 

\section{Discussion}
We have presented a proof of a tight quantum Bell inequality \revision{or Tsirelson type inequality} for the $2\times2\times\infty$ Bell setting, i.e.\ for two parties and two measurements for each side which can each have infinitely many outcomes. This quantum Bell inequality is a direct analog of Tsirelson's \revision{original} inequality \cite{Cir80} for the case of infinitely many outcomes. This analogy becomes much more obvious when writing Tsirelson's \revision{original} inequality in the form of our new tight quantum Bell inequality
\begin{eqnarray}
\!\!\!\prob{A_2<B_2}+\prob{B_2<A_1}+\prob{A_1<B_1}+\nonumber\\
+\,\prob{B_1\leq A_2}\geqslant I_Q^d.\!\!\!\label{eq:disscusion}
\end{eqnarray}
Tsirelson's inequality takes the form of \eqref{eq:disscusion} for two outcomes with $I_Q^{d=2}\!=\!(3-\sqrt{2})/2\approx 0.79$. Our inequality corresponds to the case of infinitely many outcomes with $I_Q^{d=\infty}\!=\! 0$. As in the case of Tsirelson's inequality our proof does not require any assumptions on the Hilbert space or the kind of operators involved. However, we show that the maximal violation can be achieved by a pure state whose explicit form is presented above and by the conjectured best measurement operators.

It is interesting to notice that the presented quantum Bell inequality \revision{is a tight no-signalling inequality and} is maximal in the sense that it corresponds to a face of the polytope of normalized probability vectors. A similar situation was observed in \cite{stefano} for the case of a $2\times N \times d$ Bell setting in the limit $N\to \infty$. It was seen there that this property of the quantum Bell inequality is particularly interesting in the context of quantum key distribution. While in \cite{stefano} the optimal state when $N\to \infty$ is the maximal entangled state, in our case the situation is more complicated. We leave the analysis of the relevance of the here presented quantum Bell inequality for quantum key distribution for future work. Further, one can also generalize the above quantum Bell inequality for the case of a $2\times N \times d$ Bell setting in the limit $d\to \infty$ for any $N\geq2$. Again the situation becomes slightly more complicated as the optimal state in not maximally entagled. The details of this generalization will be presented elsewhere.

Let us finally also mention that the presented quantum Bell inequality might be potentially relevant for experimental implementations through high-dimensional entangled photon pairs. The violation of Bell inequalities with spatial entanglement has recently been introduced in context of the CHSH inequality \cite{violation} (see also \cite{werner2009} for an interesting more recent account). There is hope that with respect to the inequality presented above an experimental  implementation through high-dimensional orbital-angular-momenum entanglement \cite{spatial} might be feasible.

\acknowledgments
P.R.\ and W.W.\ acknowledge support of the Marie Curie Training Network ENRAGE, MRTN-CT-2004-005616. S.Z. would like to thank Wouter Peeters and Bart-Jan Pors from the Leiden Institute of Quantum Optics and Quantum Information for discussions, the Department of Statistics at S\~ao Paulo University for kind hospitality and the ISAC program (Erasmus Mundus) for financial support. \revision{Further, we would like to thank the referees for their comments and suggestions to improve the presentation.}


\end{document}